\documentclass[12pt]{article}
\usepackage{mathrsfs}
\usepackage{dsfont}
\usepackage{amsthm}
\usepackage{mathrsfs}
\usepackage{amsmath}
\usepackage{amsfonts}
\usepackage[colorlinks,linkcolor=black,anchorcolor=blue,citecolor=blue]{hyperref}
\usepackage{amssymb, amsmath, cite}

\newtheorem{theorem}{Theorem}[section]
\newtheorem{lemma}{Lemma}[section]
\newtheorem{proposition}{Proposition}

\newcommand\be{\begin{equation}}
\newcommand\ee{\end{equation}}
\newcommand\ber{\begin{eqnarray}}
\newcommand\eer{\end{eqnarray}}
\newcommand\berr{\begin{eqnarray*}}
\newcommand\eerr{\end{eqnarray*}}
\newcommand\bea{\begin{eqnarray}}
\newcommand\eea{\end{eqnarray}}

\newcommand{\dd}{\mathrm{d}}
\newcommand\e{\mathrm{e}}
{\setlength\arraycolsep{2pt}

\begin{document}

\title{Domain Wall Solution Arising in Abelian Higgs Model \\Subject to Born-Infeld Theory of Electrodynamics}
\author{Lei Cao\\School of Mathematics and Statistics\\Henan University\\
Kaifeng, Henan 475004, PR China\\ \\$\mbox{Xiao Chen}^*$\\
 School of Mathematics and Information Science\\Henan University of Economics and Law\\
Zhengzhou, Henan 450016, PR China}
\date{}
\maketitle

\begin{abstract}
In this note we research the Abelian Higgs model subject to the Born-Infeld theory of electrodynamics for which the BPS equations can be reduced into a quasi-linear differential equation. We show that the equation exists a unique solution under two interesting boundary conditions which realize the corresponding phase transition. We construct the solution through a dynamical shooting method for which the correct shooting slope is unique. We also obtain the sharp asymptotic estimate for the solution at infinity.
\end{abstract}

\medskip
\begin{enumerate}

\item[]

\item[]
{MSC numbers(2020):} 34B40, 78A30, 81T13.

\end{enumerate}

\section{The minimization problem}\label{s1}

In this work, we concern with the Abelian Higgs model \cite{JT,NO} subject to the Born--Infeld theory \cite{Bor,Born,BL,BLF} of electrodynamics, whose Lagrangian action density reads as
\be\label{1.1}
\mathcal{L}=b^2\left(1-\sqrt{1+\frac{1}{2b^2}F_{\mu\nu}F^{\mu\nu}}\right)+\frac{1}{2}\overline{D_\mu\phi}D^\mu\phi-V(|\phi|^2),
\ee
where $b>0$ is called the Born parameter and $\phi$ is a complex--valued scalar field. $F_{\mu\nu}=\partial_\mu A_\nu-\partial_\nu A_\mu$ the electromagnetic field induced from $A_\mu$, $D_\mu\phi=\partial_\mu\phi-iA_\mu\phi$ the gauge--covariant derivative, and $V\geq 0$ a potential density function. In general, the space time is taken to be $\mathbb{R}^{n,1}$.

Particularly, in the two--dimensional, we are to derive a virial identity of the model \eqref{1.1} for static solution under the temporal gauge $A_0=0$. In this case, the Hamiltonian density of \eqref{1.1} may be calculated to be
\be\label{1.2}
\mathcal{H}=\frac{1}{\beta}\left(\sqrt{1+\frac{\beta}{2}F_{ij}^2}-1\right)+\frac{1}{2}|D_i\phi|^2+V(|\phi|^2),
\ee
where $\beta=\frac{1}{b^2}$, and the Euler--Lagrange equation of \eqref{1.1} are
\bea
D_i D_i\phi&=&2V'(|\phi|^2)\phi,\label{1.3}\\
\partial_j\left(\frac{F_{ij}}{\sqrt{1+\beta F_{12}^2}}\right)&=&\frac{i}{2}(\phi\overline{D_i \phi}-\bar{\phi}D_i\phi),\label{1.4}
\eea
where $i,j=1,2$. By exploring the rescaled fields, $\phi^\lambda(x)=\phi(\lambda x), A_i^\lambda(x)=\lambda A_i(\lambda x)$, we arrive at the following Derrick--Pohozaev type identity
\be\label{1.5}
\int_{\mathbb{R}^2} \Big(\frac{1}{\beta}\Big[\sqrt{1+\beta F_{12}^2}-1\Big]+V(|\phi|^2)\Big)\dd x=\int_{\mathbb{R}^2} \frac{F_{12}^2}{\sqrt{1+\beta F_{12}^2}}\dd x,
\ee
which is the anticipated virial identity for a critical point of the energy.

Using the identity
\be\label{1.6}
|D_1\phi|^2+|D_2\phi|^2=|D_1\phi\pm i D_2\phi|^2\pm i(D_1\phi \overline{D_2\phi}-\overline{D_1\phi}D_2\phi),
\ee
Hamiltonian density \eqref{1.2} can be written as
\ber
\mathcal{H}=&&\frac{\Big(F_{12}\pm \frac{1}{2}\sqrt{1+\beta F_{12}^2}(|\phi|^2-1)\Big)^2}{2\sqrt{1+\beta F_{12}^2}}+\frac{\Big(\sqrt{1+\beta F_{12}^2}\sqrt{1-\frac{\beta}{4} (|\phi|^2-1)^2}-1\Big)^2}{2\beta\sqrt{1+\beta F_{12}^2}}\notag\\
&&-\frac{1}{\beta}\mp\frac{1}{2}F_{12}(|\phi|^2-1)+\frac{1}{\beta}\sqrt{1-\frac{\beta}{4} (|\phi|^2-1)^2}\notag\\
&&+\frac{1}{2}|D_1\phi\pm i D_2\phi|^2\pm\frac{i}{2}(D_1\phi \overline{D_2\phi}-\overline{D_1\phi}D_2\phi)+V(|\phi|^2).\label{1.7}
\eer

Now choose
\be\label{1.8}
V(|\phi|^2)=\frac{1}{\beta}\left(1-\sqrt{1-\frac{\beta}{4} (|\phi|^2-1)^2}\right).
\ee
If $V=0$, there is a spontaneously broken $U(1)$ symmetry as in the formalism of Abelian Higgs theory. For convenience, we assume $V(1)=0$ and $|\phi|^2=\phi_0^2=1$.

Besides, note that $|\phi(x)|\rightarrow 1$ as $|x|\rightarrow\infty$, we know $D_1\phi$ and $D_2\phi$ vanish at infinity rapidly, thus
\be\label{1.9}
\int_{\mathbb{R}^2} \big(i(D_1\phi \overline{D_2\phi}-\overline{D_1\phi}D_2\phi)-|\phi|^2F_{12}\big)\dd x=0,
\ee
which leads us to the following energy lower bound
\bea\label{1.10}
E(\phi,A)&=&\int_{\mathbb{R}^2}\mathcal{H}\dd x\notag\\
&=&\int_{\mathbb{R}^2}\Bigg(\frac{\Big(F_{12}\pm \frac{1}{2}\sqrt{1+\beta F_{12}^2}(|\phi|^2-1)\Big)^2}{2\sqrt{1+\beta F_{12}^2}}+\frac{\Big(\sqrt{1+\beta F_{12}^2}\sqrt{1-\frac{\beta}{4} (|\phi|^2-1)^2}-1\Big)^2}{2\beta\sqrt{1+\beta F_{12}^2}}\notag\\
&&\pm\frac{1}{2}F_{12}+\frac{1}{2}|D_1\phi\pm i D_2\phi|^2\Bigg)\dd x\geq\pm\frac{1}{2}\int_{\mathbb{R}^2}F_{12}\dd x.
\eea
The lower bound is obtained when the following equations are satisfied
\bea
&&F_{12}\pm \frac{1}{2}\sqrt{1+\beta F_{12}^2}(|\phi|^2-1)=0,\label{1.11}\\
&&\sqrt{1+\beta F_{12}^2}\sqrt{1-\frac{\beta}{4} (|\phi|^2-1)^2}-1=0,\label{1.12}\\
&&D_1\phi\pm i D_2\phi=0.\label{1.13}
\eea
It may be examined that \eqref{1.11} implies \eqref{1.12}, these two equations can be compressed into one equation
\be\label{1.14}
F_{12}=\pm\frac{1-|\phi|^2}{2\sqrt{1-\frac{\beta}{4} (|\phi|^2-1)^2}}.
\ee

Using $u=\ln |\phi|^2$ for the equations \eqref{1.13}--\eqref{1.14}, Yang \cite{Yang} established the existence and uniqueness theory for its $N$--vortex solution subject to the boundary condition $u(\pm\infty)=0$ corresponding to $|\phi(\pm\infty)|^2=1$. As far as we know, there are no results for non-topological boundary conditions.  In particular, when $\beta=0$, we get the classical Abelian Higgs theory \cite{JT}, we may also refer to the Taubes equation \cite{Ma}.

The problem \eqref{1.11}--\eqref{1.13} is of outstanding interest in one--dimensional which produces domain walls. We now pursue a domain--wall structure contained in the Abelian Higgs model subject to the Born--Infeld theory of electrodynamics. For this purpose, we assume the fields $\phi$ and $A_i$ only depend on $x^1=x$, and take the ansatz \cite{Bolo}
\be\label{1.15}
A_1=0,~~A_2=a(x),~~\phi=f(x)=\mbox{real}.
\ee
Then \eqref{1.13} becomes $f'\pm af=0$. In the nontrivial situation, $f$ never vanishes. Without loss of generality, we may assume $f>0$. Therefore, we see that $a=\mp(\ln f)'$. Moreover, we have $F_{12}=a'$. Substituting these into equation \eqref{1.14} and let $u=2\ln f$, we get the self-dual domain--wall equation
\be\label{2.1}
u''=\frac{\e^u-1}{\sqrt{1-\frac{\beta}{4}(\e^u-1)^2}}.
\ee
Following \cite{Bolo}, boundary conditions of interest describing relevant phase transition phenomena include Higgs to magnetic phase
\be\label{2.2}
u(-\infty)=0,\quad u(\infty)=-\infty,
\ee
and magnetic to magnetic phase
\ber\label{2.3}
u(-\infty)=u(\infty)=-\infty.
\eer
When $\beta=0$, equation \eqref{2.1} is a one--dimensional Liouville type equation \cite{Lio} and we have studied in \cite{Cao}.

In this work, we shall show the existence and uniqueness of equation \eqref{2.1} with boundary conditions \eqref{2.2} and \eqref{2.3} respectively. We will also obtain the sharp asymptotic estimate for the solution $u(x)$ at infinity.

\section{Existence and uniqueness of an domain wall}
\setcounter{equation}{0}

We are going to solve the two--point boundary value problems consisting of \eqref{2.1} and two different boundary conditions \eqref{2.2} and \eqref{2.3} over the full interval $(-\infty,\infty)$. For this purpose, we approach a dynamical shooting method.

Our main results are stated as follows.

\begin{theorem}\label{th2.1}
Suppose the parameter $\beta<4$, the two--point boundary value problem consisting of equation \eqref{2.1} and two relevant boundary conditions \eqref{2.2} and \eqref{2.3} over the interval $(-\infty,\infty)$  has a unique solution
$u(x)$ which enjoys the following properties.

\begin{enumerate}
\item[(i)] Under the boundary condition \eqref{2.2}, solution $u$ strictly decreases such that
$u(x)<0$ for $x\in (-\infty,\infty)$ and enjoys the sharp boundary estimates given by \label{2.20}
\ber
u(x)&=&-\frac{x^2}{2}+\mbox{\rm O}(1),~~x\to \infty,\label{2.15}\\
u(x)&=&\mbox{\rm O}(\e^{-|x|}),~~x\to -\infty.\label{2.18}
\eer

\item[(ii)] Under the boundary condition \eqref{2.3}, for any given point $x_0$ with $u(x_0)\leq 0$, $u$ attains its unique global maximal value $u_0$ at $x=x_0$, which is given by the first-order equation \eqref{2.19}, and enjoys the sharp asymptotic behavior
    \be
    u(x)=-\frac{x^2}{2}+\mbox{\rm O}(1),~~x\to \pm\infty.
    \ee
\end{enumerate}
\end{theorem}

Obviously, any solution $u$ of function \eqref{2.1} with boundary conditions \eqref{2.2} or \eqref{2.3} must satisfy $u\leq 0$. Otherwise, we may assume $u>0$ somewhere then there exists a point $x_0\in\mathbb{R}$ where $u$ attain its local maximum in $\mathbb{R}$ and $u''(x_0)\leq 0$, which is impossible in view of \eqref{2.1}.

In particular, any solution $u$ of \eqref{2.1} satisfying \eqref{2.2} must be negative-valued. In fact, note that $u(-\infty)=0$, if $u=0$ at a point $x_1\in\mathbb{R}$, then $u\equiv 0$ on $(-\infty,x_1)$ or $u<0$ at some point belongs to $(-\infty,x_1)$. For the first case, taking any point $x_2\in (-\infty,x_1)$, clearly $u'(x_2)=0$ and $u(x_2)=0$. Applying the uniqueness theorem for the initial value problems of an ordinary differential equation we get $u\equiv 0$ throughout $\mathbb{R}$ which violates the condition $u(\infty)=-\infty$. For the later case $u<0$, we can obtain a minimum at some point $x_3\in(-\infty,x_1)$, then $u''(x_3)\geq 0$ contradicts with \eqref{2.1}.

Further more, there must hold $1-\frac{\beta}{4}(\e^u-1)^2>0$ to make sense of the right-hand side of equation \eqref{2.1}. In other words function \eqref{2.1} with boundary conditions \eqref{2.2} or \eqref{2.3} have solutions only when $\beta< 4$.

We first study the function \eqref{2.1} with boundary \eqref{2.2}. In order to solve this problem, considering the initial value problem
\be\label{2.4}
u''=\frac{\e^u-1}{\sqrt{1-\frac{\beta}{4}(\e^u-1)^2}},\quad u(0)=a,\quad u'(0)=-b,\quad a,b>0.
\ee
where $a,b$ are constants and $b$ is an initial slope. Note the autonomous of the equation, the choice of the initial point $x=0$ is arbitrary. We shall use a dynamical shooting method to approach the problem.

\begin{proposition}\label{pro2.1}
For any fixed $a>0$ there exists a unique $b>0$ such that the unique solution of the initial value problem \eqref{2.4} solves the boundary value problem \eqref{2.1}--\eqref{2.2}.
\end{proposition}

We first concentrate on the half $x<0$. Now for convenience we let $t=-x$. Then in the region $x<0$ the system \eqref{2.4} becomes
\be\label{2.5}
u''=\frac{\e^u-1}{\sqrt{1-\frac{\beta}{4}(\e^u-1)^2}},\quad u(0)=a,\quad u'(0)=b,\quad a,b>0.
\ee

In order to prove the theorem, we define the following shooting sets
\ber
{\mathcal{B}}^{-}&=&\big\{b\geq 0~|~u'(t)<0~\text{for some}~t\geq 0 \big\},\notag\\
{\mathcal{B}}^{0}&=&\big\{b\geq 0~|~u'(t)>0\mbox{ and } u(t)<0~\text{for all}~t\geq 0 \big\},\notag\\
{\mathcal{B}}^{+}&=&\big\{b\geq 0~|~u'(t)>0~\text{for all}~t\geq 0\mbox{ and } u(t)>0~\text{for some}~t\geq 0\big\},\notag
\eer
in which, the solution $u(x)$ need not to exist for all $x$, the statement are made to mean wherever the solution exists.

\begin{lemma}\label{d.1}
There hold
${\mathcal{B}}^{+}\cap{\mathcal{B}}^{-}={\mathcal{B}}^{+}\cap{\mathcal{B}}^{0}
={\mathcal{B}}^{-}\cap{\mathcal{B}}^{0}=\varnothing$ and $[0,\infty)={\mathcal{B}}^{+}\cup{\mathcal{B}}^{0}\cup{\mathcal{B}}^{-}$.
\end{lemma}
\begin{proof}
Clearly, $\mathcal{B}^{-}$, $\mathcal{B}^{0}$,  and $\mathcal{B}^{+}$ are disjoint.
If $b\geq 0$ but $b\notin\mathcal{B}^{-}$, then $u'(t)\geq0$ for all $t\geq 0$.  Assume there is a point $t_0\geq 0$ so that $u'(t_0)=0$. Then $\e^{u(t_0)}-1\neq 0$, otherwise $u(t_0)=0$ which means $u$ arrives
at an equilibrium of the differential equation at $t_0$, which violates the uniqueness theorem for solutions to initial value problems of ordinary differential equations.
So we have $u''(t_0)>0$ or $u''(t_0)<0$ in view of $u'(t_0)=0$ and $\e^{u(t_0)}-1\neq 0$ at $t=t_0$. Therefore,
there holds $u'(t)<0$ for $t$ close to $t_0$ but $t<t_0$ if $u''(t_0)>0$ or $t>t_0$ if $u''(t_0)<0$, contradicting the assumption $b\notin\mathcal{B}^{-}$. Thus, $b\notin\mathcal{B}^{-}$ implies $u'(t)>0$ for all $t$. In other words, $b\in \mathcal{B}^{0}\cup\mathcal{B}^{+}$, therefore, $[0,\infty)={\mathcal{B}}^{+}\cup{\mathcal{B}}^{0}\cup{\mathcal{B}}^{-}$.
\end{proof}

\begin{lemma}\label{d.2}
The sets $\mathcal{B}^{+}$ and $\mathcal{B}^{-}$ are both open and nonempty.
\end{lemma}
\begin{proof}
We first show that $\mathcal{B}^{+}$ is open and nonempty. Integrating the equation in \eqref{2.5} gives us
\ber
u'(t)&=&b+\int_0^t\frac{\e^{u(\tau)}-1}{\sqrt{1-\frac{\beta}{4}(\e^{u(\tau)}-1)^2}}\dd\tau,\label{2.6}\\
u(t)&=&-a+bt+\int_0^t\int_0^s\frac{\e^{u(\tau)}-1}{\sqrt{1-\frac{\beta}{4}(\e^{u(\tau)}-1)^2}}\dd\tau\dd s.\label{2.7}
\eer

Suppose $u(t)=\mbox{\rm O}(\e^{-t})\mbox{ as }t\to \infty$ (we will show the rationality of the hypothesis in proposition \ref{pro2.2}), then the integral on the right-hand-side of \eqref{2.6} and \eqref{2.7} converge as $t$ and $s$ approach $\infty$. Thus for any fixed $t>0$ we can choose $b>0$ sufficiently large such that $u(t)>0$.

Next to show that $u(t)$ is increasing. $u(t)$ is increasing in the neighbour of $0$ since $b\geq 0$. Suppose there exists a point $t_2$ so that $u'(t_2)<0$. Then if $u(t_2)>0$, $u(t)$ will obtain a positive local maximum on $(0,x_2)$, violating the equation in \eqref{2.5}. Thus $u(t_2)<0$ is the only possible situation. Obviously, $u(t)$ is negative-valued for all $t\in [0,t_2]$, otherwise we would get another nonnegative local maximum which contradicts again the equation in \eqref{2.5}.

Note that $u(t_1)>0$, we get $t_1>t_2$, combine with the assumption $u'(t_2)<0$ we know $u$ has a negative local minimum in $(t_2,t_1)$. This is impossible in view of the equation in \eqref{2.5}.

Therefore $u(t)$ is increasing for all $t\in [0,\infty)$.

It remains to show that $u(t)$ is strictly increasing. Let $t_3<t_4$ be such that $u(t_3)=u(t_4)$. Then $u(t)=u(t_3)$ for all $t\in (t_3,t_4)$. Then by the uniqueness theorem for solutions to initial value problems of ordinary differential equation we get $u\equiv 0$ for all $t\geq 0$, which is false. Thus $u'(t)>0$ for all $t\in [0,\infty)$.

To show that $\mathcal{B}^{+}$ is open, let $b_0\in \mathcal{B}^{+}$ and use $u(t;b_0)$ to denote the corresponding solution of \eqref{2.5} so that $u'(t;b_0)>0$ for all $t\geq 0$ and $u(t_0;b_0)>0$ for some $t_0>0$. By the continuous dependence theorem for the solutions of the initial value problems of ordinary differential equation, there exists a neighborhood of $b_0$, denote as $U(b_0;\delta_1)(\delta_1>0)$, to make that for any $b\in U(b_0;\delta)$, the solution of \eqref{2.5}, say $u(t,b)$, satisfies $u'(t;b)>0$ for all $t\geq 0$ and $u(t_0;b)>0$ in view of \eqref{2.6}--\eqref{2.7}. Thus $b_0$ is an interior point of $\mathcal{B}^{+}$.

For the statement concerning $\mathcal{B}^{-}$. Let $b=0$, for $t>0$ sufficiently small, we have $u''(t)<0$ from \eqref{2.5}. Thus $u'(t)$ is decreasing and $u'(t)<0$. Let $t_5$ be the point satisfies $u'(t_5)>0$
and $u(t_5)<0$, so there is a point $t_6\in (0,t_5)$ such that $t_6$ is a local minimum of $u(t)$ and $u(t_6)<0$, which violates the equation in \eqref{2.5}. Therefore $u'(t)\leq 0$ for all $t\geq 0$. In fact $u(t)$ is strictly decreasing for all $t>0$. The conclusion can be reached by using the same method as in $\mathcal{B}^{+}$. This proves ${0}\in \mathcal{B}^{-}$, hence $\mathcal{B}^{-}$ is nonempty.

Finally, we need to show the openness to $\mathcal{B}^{-}$. Because of $u'(0)=b>0$, there is a point $t_7>0$ so that $u'(t_7)$=0 and $u''(t_7)<0$. Clearly $u(t_7)<0$, otherwise it violates the equation in \eqref{2.5}. Since there can not be any negative local minimum point, we get that
\be\notag
u'(t)>0,0<t<t_7~~\mbox{ and }~~ u'(t)<0,t>t_7.
\ee
Let $t_8>t_7$, then $u'(t_8;b)<0$. In terms of the continuous dependence of $u(t_8;b)$ on $b$, we can find a neighborhood of $b$, say $U(b;\delta_2)(\delta_2>0)$, so that for any $\widetilde{b}\in U(b;\delta_2)$, there establish $u'(t_8;\widetilde{b})<0$. We arrive at the conclusion that $\widetilde{b}\in \mathcal{B}^{-}$.

The proof of the lemma is complete.
\end{proof}

\begin{lemma}\label{d.3}
The set $\mathcal{B}^{0}$ is closed and nonempty.
\end{lemma}
\begin{proof}
This consequence is a direct conclusion note the connectedness of $[0,\infty)$.
\end{proof}

\begin{lemma}\label{d.4}
There is only one point in $\mathcal{B}^{0}$.
\end{lemma}
\begin{proof}
Let $b_1,b_2\in \mathcal{B}^{0}$, $u(t;b_1)$ and $u(t;b_2)$ are two solutions of \eqref{2.5}, then $w(t)=u(t;b_1)-u(t;b_2)$ satisfies the boundary condition $w(0)=w(\infty)=0$, and the equation
\be\notag
w''(t)=\frac{\e^\xi}{\big(1-\frac{\beta}{4}(\e^\xi-1)^2\big)^{\frac{3}{2}}}w(t),
\ee
where $\xi$ lies between $u(t;b_1)$ and $u(t;b_2)$. Hence the maximum principle shows that $w(t)=0$ everywhere. Therefore, we have $b_1=b_2$.
\end{proof}

\begin{lemma}\label{d.5}
Let $b\in\mathcal{B}^{0}$, we have $u(t)\to 0$ as $t\to\infty$.
\end{lemma}
\begin{proof}
Since $b\in\mathcal{B}^{0}$, there is $u'(t)>0$ and $u(t)<0$ for all $t\geq 0$. Thus there exists the limit of $u(t)$, say $u_0$, as $t\to\infty$, and $-\infty<u_0\leq 0$. If $u_0<0$, then we have $u(t)<u_0$ for all $t>0$. Note the equation in \eqref{2.5}, we can find $t_0>0$ sufficiently large so that $u''(t)<-\varepsilon<0$ for any $\varepsilon>0$ and $t\geq t_0$. In particular, $u'(t)<u'(t_0)-\varepsilon(t-t_0)$ for all $t>t_0$. It will lead to a contradiction when $t$ is sufficiently large because $u'(t)>0$ for all $t\geq 0$. Therefore $u(t)\to 0$ as $t\to\infty$.
\end{proof}

Returning to the original variable $x=-t$, we arrive at the following expression for the solution of \eqref{2.4}
\be\notag
u(x)\to 0, x\to -\infty;~~~u'(x)<0, u(x)<0,-\infty<x\leq 0.
\ee
Since the initial value $a,b>0$ and $u''(x)<0$ for all $x>0$. These properties make $u(x)<0, u'(x)<0$ and $u''(x)<0$ establish when $x>0$. Such a fact leads to a direct consequence
\be\label{2.8}
\lim_{x\to\infty}u(x)=\lim_{x\to\infty}u'(x)=\lim_{x\to\infty}u''(x)=-\infty.
\ee
In fact, we have the following more accurate estimates.

\begin{proposition}\label{pro2.2}
The solution of \eqref{2.4} has the sharp decay estimates $u(x)=-\frac{x^2}{2}+\mbox{\rm O}(1), x\to \infty$ and $u(x)=\mbox{\rm O}(\e^{-|x|}), x\to -\infty$.
\end{proposition}
\begin{proof}
In view of the boundary condition $u(-\infty)=0$ we get $u'(-\infty)=0$. Multiplying the equation \eqref{2.1} by $u'$ and integrating over $(-\infty,x)$, we arrive at
\be\label{2.9}
(u')^2=\int_{0}^{u(x)}\frac{2(\e ^u-1)\dd u}{\sqrt{1-\frac{\beta}{4}(\e ^u-1)^2}}.
\ee
For any $u<0$, there holds $-1<\e ^u-1<0$, so we have
\be\label{2.10}
\frac{2(\e ^u-1)}{\sqrt{4-\beta}}<\frac{(\e ^u-1)}{\sqrt{1-\frac{\beta}{4}(\e ^u-1)^2}}<(\e ^u-1).
\ee
Inserting \eqref{2.10} into \eqref{2.9}, we obtain
\be\label{2.11}
2(\e ^u-u-1)<u'^2<\frac{4(\e ^u-u-1)}{\sqrt{4-\beta}}.
\ee
Let us first consider the behavior of $u(x)$ as $x\to \infty$. Note that $u'(x)<0$ for all $x\in (-\infty,\infty)$, we can rewrite \eqref{2.11} as
\be\notag
-\frac{2\sqrt{\e ^u-u-1}}{(4-\beta)^{\frac{1}{4}}}<u'<-\sqrt{2(\e ^u-u-1)}.
\ee
Suppose $x_0$ be the point such that $u_0=u(x_0)<-1$. Therefore, we get the integral
\be\label{2.12}
-\frac{2(x-x_0)}{(4-\beta)^{\frac{1}{4}}}<\int_{u_0}^{u(x)}\frac{\dd u}{\sqrt{\e ^u-u-1}}<-\sqrt 2(x-x_0),~~x>x_0.
\ee
Using $-u-1<\e ^u-u-1<-u$ in \eqref{2.12} we have
\be\label{2.14}
\sqrt{-u}-\sqrt{-u_0}<\frac{x-x_0}{(4-\beta)^{\frac{1}{4}}}~~\mbox{and}~~\sqrt{2}(x-x_0)<2(\sqrt{-u-1}-\sqrt{-u_0-1}),
\ee
which lead to the sharp asymptotic estimate \eqref{2.15}.

In a similar way, we can estimate the behavior of $u(x)$ as $x\to -\infty$. For this reason, using the fact $u(-\infty)=0$ and $u'(x)<0$, we can choose $x_0$ sufficiently negative such that $\e^{u(x)}>1-\varepsilon$, for any $\varepsilon\in(0,1)$ and $x<x_0$. Hence, we get the inequality
\be\label{2.16}
(1-\varepsilon)u^2(x)<2\left(\e^{u(x)-u(x)-1}\right)<u^2(x),~~x<x_0.
\ee
Inserting this into \eqref{2.12}, we arrive at
\be\label{2.17}
\frac{(4-\beta)^{\frac{1}{4}}}{\sqrt{2(1-\varepsilon)}}\ln\left|\frac{u}{u_0}\right|<x-x_0<\ln\left|\frac{u}{u_0}\right|,
\ee
Since $\varepsilon>0$ may be arbitrarily small, we deduce the sharp asymptotic estimate \eqref{2.18}.
\end{proof}

In view of the study of the equation \eqref{2.1} with boundary \eqref{2.2}, we may construct solutions of the equation \eqref{2.1} satisfying boundary condition \eqref{2.3}. In fact, such a solution may have a global maximum point. By translation invariance of \eqref{2.1}, we assume the maximum point of $u(x)$ is $0$ and $u(0)=u_0$, so $u''(0)\leq 0$. In view of \eqref{2.1}, we have $u_0\leq 0$. Consequently, $u(x)\leq 0$ for all $x$. Therefore, using $u'(0)=0$ and multiplying \eqref{2.1} by $u'$, integrating around $x=0$, we are lead to
\be\label{2.19}
(u')^2=\int_{u_0}^{u(x)}\frac{2(\e ^u-1)\dd u}{\sqrt{1-\frac{\beta}{4}(\e ^u-1)^2}}.
\ee
Since the right-hand side of \eqref{2.19} decreases about $u$, thus the integral is positive for all $u<u_0$. In other words, the right-hand side of \eqref{2.19} keeps positive for $x\neq 0$. Note the boundary condition $u(\infty)=-\infty$, so we get the inequality
\be\label{2.20}
-\frac{2x}{(4-\beta)^{\frac{1}{4}}}<\int_{u_0}^{u(x)}\frac{\dd u}{\sqrt{\e ^u-u-1}}<-\sqrt 2x,~~x>0.
\ee
For the part in $x<0$, since the solution goes to $-\infty$ as $x\to-\infty$, we may flip the solution in $x>0$ by setting $x\mapsto-x$ to obtain the solution with $x<0$. In this way, we get a solution of \eqref{2.1} which satisfies the boundary condition \eqref{2.3}.

\medskip

In conclusion, we have shown that the Abelian Higgs model subject to the Born--Infeld theory of electrodynamics has a unique domain wall solution which may be obtained by finding a correct initial slope of the solution to the initial value problem associated with the quasi-linear differential equation. This implies there is only one way to achieve phase transition between the superconducting and normal states corresponding to the boundary condition $f(-\infty)=1$ and $f(\infty)=0$ and between the normal and normal states corresponding to the boundary condition $f(-\infty)=f(\infty)=0$. Furthermore, some asymptotic estimates for the solution are also obtained. It is worth noting that when the function \eqref{2.1} satisfies boundary condition \eqref{2.2}, we get the estimate \eqref{2.18} which confirms our hypothesis in Lemma \ref{d.2} is reasonable.

\medskip

{\bf Acknowledgments.}
The authors would like to thank Professor Yisong Yang for bringing us to this problem. This work was supported by NSFC-12101197 and the 72nd batch of China Postdoctoral Science Foundation. No. 2022M721022.

\end{document}